\documentclass[letterpaper, 10 pt, conference]{ieeeconf}
\IEEEoverridecommandlockouts
\usepackage{graphicx}
\usepackage{amsmath,amssymb}
\usepackage{flushend}
\usepackage{color}
\usepackage{latexsym}\usepackage{comment}
\usepackage{url}
\usepackage{graphics}
\usepackage{pdfsync}
\usepackage{stfloats}
\usepackage{epsfig}
\usepackage{times}
\usepackage{amssymb}
\usepackage{amsmath}
\usepackage{amsfonts}
\usepackage{mathtools}
\usepackage{textcomp}
\usepackage{xcolor}
\usepackage{indentfirst}
\usepackage{epstopdf}
\usepackage{multirow}
\usepackage{float}
\usepackage{color}
\usepackage{enumerate}
\usepackage{subfigure}
\usepackage{comment}
\usepackage{cite}
\usepackage{caption}
\usepackage{bm}

\usepackage[ruled, vlined]{algorithm2e}
\usepackage{algorithmic}

\DeclareMathOperator*{\minimize}{minimize}

\usepackage{tcolorbox}

\title{\LARGE \bf Data-Driven Min-Max MPC for LPV Systems with Unknown Scheduling Signal}

\author{Yifan~Xie, Julian~Berberich, Felix~Br\"{a}ndle, Frank Allg\"{o}wer
	\thanks{F. Allg\"{o}wer is thankful that his work was funded by Deutsche Forschungsgemeinschaft (DFG, German Research
Foundation) under Germany’s Excellence Strategy - EXC 2075 - 390740016 and under grant 468094890.
F. Allg\"{o}wer acknowledges the support by the Stuttgart
Center for Simulation Science (SimTech).
The authors thank the International Max Planck Research School
for Intelligent Systems (IMPRS-IS) for supporting Yifan Xie and Felix Br\"{a}ndle.
	}
	\thanks{Yifan~Xie, Julian~Berberich, Felix~Br\"{a}ndle, Frank Allg\"{o}wer are with the Institute for Systems Theory and Automatic Control, University of Stuttgart, 70550 Stuttgart, Germany.
		{\tt\small  (email: \{yifan.xie, julian.berberich, felix.braendle, frank.allgower\}@ist.uni-stuttgart.de}). }
}

\newtheorem{mythm}{Theorem}

\newtheorem{mylem}{Lemma}

\newtheorem{myexm}{Example}
\newtheorem{remark}{Remark}
\newtheorem{assum}{Assumption}

\begin{document}
	
	\maketitle

\begin{abstract}
This paper presents a data-driven min-max model predictive control (MPC) scheme for linear parameter-varying (LPV) systems.
Contrary to existing data-driven LPV control approaches, we assume that the scheduling signal is unknown during offline data collection and online system operation.
Assuming a quadratic matrix inequality (QMI) description for the scheduling signal, we develop a novel data-driven characterization of the consistent system matrices using only input-state data.
The proposed data-driven min-max MPC minimizes a tractable upper bound on the worst-case cost over the consistent system matrices set and over all scheduling signals satisfying the QMI.
The proposed approach guarantees recursive feasibility, closed-loop exponential stability and constraint satisfaction if it is feasible at the initial time.
We demonstrate the effectiveness of the proposed method in simulation.
\end{abstract}

\section{Introduction}\label{sec:1}
Using data to design or improve controllers is becoming increasingly important with the growing complexity of engineering systems. With this motivation, various frameworks have been proposed to design data-driven controllers, e.g., data informativity \cite{van2023informativity} and Willems' Fundamental Lemma \cite{willems2005note,berberich2024overview,markovsky2021behavioral,markovsky2023datapower}. 
While these methods were initially focused on linear time-invariant systems, they have been extended to classes of nonlinear systems \cite{martin2023guarantees,berberich2022linear}.
Linear parameter-varying (LPV) systems are linear systems involving a time-varying scheduling signal and they provide a powerful framework for addressing nonlinear analysis and control problems using linear methods \cite{toth2010modeling}.

Several data-driven controller design methods have been proposed to stabilize LPV systems using noise-free data \cite{verhoek2022direct} or noisy data \cite{mejari2023data,miller2022data,verhoek2024decoupling}. 
Performance metrics such as infinite horizon quadratic cost \cite{verhoek2022direct},  
$\mathcal{H}_2$ norm  \cite{miller2022data,verhoek2022direct}, and $\ell_2$ gain \cite{verhoek2022direct} are considered in these approaches. 
Some of them rely on a characterization of all LPV systems consistent with measurements of input, state, and parameter signals, and they design a gain-scheduling controller to guarantee stability \cite{miller2022data,verhoek2024decoupling}.
A robust data-driven controller design for linear system with nonlinear uncertainties is proposed in \cite{berberich2023combining}, which encompasses LPV systems as a special case.
Additionally, data-driven model predictive control (MPC) for LPV systems has been developed based on Willems' Fundamental Lemma \cite{verhoek2023linear}.
However, all of the aforementioned works assume that scheduling signal is known during offline data collection, which may not always hold in practical data-driven settings.

We propose a data-driven min-max MPC scheme for LPV systems using only a sequence of input-state data. 
Throughout the paper, we assume that the system matrices and the scheduling signal of the LPV system are unknown, but the latter satisfies a known quadratic matrix inequality (QMI). 
Based on this QMI, we characterize all system matrices explaining the available input-state data.
The proposed scheme guarantees that all systems consistent with the data are exponentially stabilizes and satisfy the constraints for all scheduling signals satisfying the QMI.
A numerical example shows the effectiveness of the proposed method.
The proposed approach is inspired by existing min-max MPC schemes \cite{kothare1996robust,lu2000quasi,morato2020model}, which only consider the case of known system dynamics.
Moreover, since our min-max MPC relies on the repeated computation of a robust feedback law, we also solve the open problem of data-driven feedback design for LPV systems with unknown scheduling signals.

The problem setup is introduced in Section \ref{sec:2}.
The data-driven characterization of the system matrices is proposed in Section \ref{sec:3}.
Based on this characterization, the data-driven min-max MPC scheme is presented in Section \ref{sec:4}.
We prove that the MPC scheme is recursively feasible and guarantees closed-loop exponential stability and constraint satisfaction.
The numerical example is presented in Section \ref{sec:5} and the paper is concluded in Section \ref{sec:6}.

\textit{Notation:} We denote the set of integers in the interval $[a, b]$ by $\mathbb{I}_{[a, b]}$ , the set of natural numbers by $\mathbb{N}$, and the identity matrix by $I$.
For a matrix $P$, we write $P\succ 0$ if it is positive and $P\succeq 0$ if it is positive-definite. 
For a vector $x$ and a matrix $P\succeq 0$, we write $\|x\|_P^2=x^\top P x$.
For matrices $A, B$ of compatible dimensions, we abbreviate $A^\top BA$ by $[\star]^\top BA$.

\section{Problem Setup}\label{sec:2}
We consider a discrete-time LPV system
\begin{equation}\label{system}
\begin{aligned}
x_{t+1}&=A_sx_t+B_su_t+\omega_t,\\
z_t&=Cx_t+Du_t,\\
\omega_t&=\Delta_tz_t,
\end{aligned}
\end{equation}
where $x_t\in\mathbb{R}^{n_x}$ is the state, and $u_t\in\mathbb{R}^{n_u}$ is the control input.
The scheduling signal $\Delta_t\in\mathbb{R}^{n_x\times n_z}$ enters the system dynamics from $z_t$ to $\omega_t$, where $z_t\in\mathbb{R}^{n_z}$ and $\omega_t\in\mathbb{R}^{n_x}$.
We assume that the matrices $A_s\in\mathbb{R}^{n_x\times n_x}$ and $B_s\in\mathbb{R}^{n_x\times n_u}$ are unknown.
The matrices $C\in\mathbb{R}^{n_z\times n_x}$ and $D\in\mathbb{R}^{n_z\times n_u}$ are assumed to be known and they express a form of prior knowledge on which states and inputs are affected by $\Delta_t$.
Moreover, $\Delta_t$ is not available, contrary to the existing literature on data-driven control for LPV systems \cite{verhoek2022direct,mejari2023data,miller2022data,verhoek2024decoupling,berberich2023combining}.

\begin{remark}\upshape
LPV systems provide a powerful framework for modeling nonlinear systems by viewing nonlinear components as scheduling signals \cite{toth2010modeling}.
In this case, when the nonlinearity is unknown, the corresponding scheduling signal $\Delta_t$ is also unknown, motivating the considered setup.
Furthermore, $\Delta_t$ can also be interpreted as a multiplicative noise term, whereas existing works on data-driven control discussed in Section \ref{sec:1} focus on additive noise. 
We assume for simplicity that $w_t=\Delta_tz_t$ enters the dynamics directly. Including additional structure, e.g., considering $w_t=B_\Delta \Delta_t z_t$ with a known matrix $B_\Delta$, is an interesting future research direction, see, e.g., \cite{berberich2023combining} for analogous results for LTI systems.
\end{remark}

We use the following example to demonstrate the flexibility of the system class \eqref{system}.
\begin{myexm}\upshape
The unbalanced disc system is a common example from the LPV literature \cite{verhoek2023linear} and is modeled as
\begin{equation}\nonumber
\begin{aligned}
\begin{bmatrix}
x_{1,t+1}\\
x_{2,t+1}
\end{bmatrix}\!\!=\!\!
\begin{bmatrix}
1\!\!-\!\!\frac{T_s}{\tau} \!\!\!\!&0\\
1 \!\!\!\!&T_s\end{bmatrix}\!\!\!
\begin{bmatrix}
x_{1,t}\\
x_{2,t}
\end{bmatrix}\!\!+\!\!\begin{bmatrix}
\frac{T_sK_m}{\tau}\\ 0
\end{bmatrix}u_t\!\!+\!\!
\begin{bmatrix}
    -\frac{T_smgl\sin(x_{2,t})}{J}\\0
\end{bmatrix},
\end{aligned}
\end{equation}
where $x_{1,t}$ is the angular velocity of the disc and $x_{2,t}$ is the angular position of the disc.
This system can be formulated as in \eqref{system} by defining the unknown matrices
\[A_s\!=\!\begin{bmatrix}
1\!-\!\frac{T_s}{\tau} &0\\
1 &T_s\end{bmatrix}, B_s\!=\!\begin{bmatrix}
\frac{T_sK_m}{\tau}\\ 0
\end{bmatrix}, \Delta_t\!=\!\begin{bmatrix}
    -\frac{T_smgl\sin(x_{2,t})}{Jx_{2,t}}\\0
\end{bmatrix}\] 
and the known matrices $C=\begin{bmatrix}
    0 &1
\end{bmatrix}$ and $D=0$. 
The scheduling signal $\Delta_t$ is unknown when the nonlinearity is not precisely known.
\end{myexm}


Throughout this paper, we assume that a sequence of input-state measurements $(U, X)$ of length $T$ is available.
Using $(U, X)$ and the available matrices $C$ and $D$, we can calculate a sequence of $z_i^d$, $\forall i\in\mathbb{I}_{[0, T-1]}$.
The available measurements are then collected in the following matrices
\vspace{-15pt}
\begin{subequations}\label{data}
\begin{align}
U&:=\begin{bmatrix}u_0^d \!\!&\ldots \!\!&u_{T-1}^d\end{bmatrix},
X:=\begin{bmatrix}x_0^d \!\!&\ldots \!\!&x_{T-1}^d \!\!&x_T^d\end{bmatrix},\label{data:input-state}\\
Z&:=\begin{bmatrix}Cx_0^d+Du_0^d &\ldots &Cx_{T-1}^d+Du_{T-1}^d\end{bmatrix}.\label{data:z}
\end{align}
\end{subequations}
The corresponding scheduling signals are denoted by $\Delta_i^d$ for all $i\in\mathbb{I}_{[0, T-1]}$, which are not available to us.

The goal of this paper is to construct a time-varying state-feedback controller $u_t=F_tx_t$ only using the collected data in \eqref{data}.
This controller should ensure that, for all scheduling signals $\Delta_t$ satisfying a certain bound (Assumption~\ref{assumption:delta} below) and all system consistent with the data, the controlled system \eqref{system} is exponentially stabilized to the origin and satisfies the constraints.
To stabilize the origin and quantify the closed-loop performance, we define a quadratic stage cost function $l(u, x)=\|x\|_Q^2+\|u\|_R^2$, where $Q, R\succ 0$.
We consider ellipsoidal input and state constraints in the data-driven min-max MPC scheme
$\|u_t\|_{S_u}\leq 1, \|x_t\|_{S_x}\leq 1, \forall t\in\mathbb{N}$, where $S_u\succ 0, S_x\succeq 0$.
In order to obtain rigorous guarantees despite the unknown $\Delta_t$, we assume that the following bound holds during both offline data generation and online operation.

\begin{assum}\upshape\label{assumption:delta}
The scheduling signals satisfy $\Delta_i^d\in\Pi, \forall i\in\mathbb{I}_{[0, N-1]}$ and $\Delta_t\in\Pi, \forall t\in\mathbb{N}$, where
\begin{equation}\label{assum:delta}
\Pi=\left\{\Delta:
\begin{bmatrix}
I \\ \Delta^\top
\end{bmatrix}^\top 
G
\begin{bmatrix}
I \\ \Delta^\top
\end{bmatrix}\succeq 0\right\}
\end{equation}
and $G=\begin{bmatrix}G_{11} &G_{12}\\G_{12}^\top &G_{22}\end{bmatrix}$ is a known invertible matrix with $G_{22}\prec 0$ and $G_{11}-G_{12}G_{22}^{-1}G_{12}^\top\succ 0$.
\end{assum}

For example, if a bound $\sigma$ on the maximum singular value of $\Delta_t$ is known, then $G$ can be chosen as $G_{11}=I, G_{12}=0, G_{22}=-\sigma^{-2} I$.

\section{Data-Driven System Characterization}\label{sec:3}

In this section, we present a data-driven characterization of the unknown system matrices $A_s$ and $B_s$ using the collected input-state measurements in \eqref{data}.

The set of all pairs $(A, B)$ consistent with the input-state data in \eqref{data} and the bound on the scheduling signal in Assumption \ref{assumption:delta} is defined as 
\[\Sigma=\{(A, B):(A, B)\in\Sigma_i, \forall i\in\mathbb{I}_{[0, T-1]}\},\]
where 
\begin{equation}\nonumber
\Sigma_i\!=\!\{\!(A, B):\!
    \exists \Delta_i^d\in\Pi  \text{ s.t. } 
    x_{i+1}^d=Ax_i^d+Bu_i^d+\Delta_i^d z_i^d\}.
\end{equation}

We first state the following technical lemma that will be used in the data-driven characterization of the set $\Sigma$.
\begin{mylem}\upshape\label{lemma1}
Suppose $E\in\mathbb{R}^{n_E \times n_x}$ has full row rank and
$\tilde{G}=\begin{bmatrix}\tilde{G}_{11} &\tilde{G}_{12}\\\tilde{G}_{12}^\top &\tilde{G}_{22}\end{bmatrix}$
satisfies 
$\tilde{G}_{22}\prec0$ and $\tilde{G}_{11}-\tilde{G}_{12} \tilde{G}_{22}^{-1}\tilde{G}_{12}^\top \succ 0$, then the following two sets are equal
\begin{equation}\label{lemma1:equ1}
S_{E\tilde{\Delta}}\coloneq
\left\{ E\tilde{\Delta} :
\begin{bmatrix} I \\ \tilde{\Delta} \end{bmatrix}^\top
\tilde{G}
\begin{bmatrix} I \\ \tilde{\Delta} \end{bmatrix}
\succeq 0
\right\}
\end{equation}
\begin{equation}\label{lemma1:equ2}
S_{\hat{\Delta}}\coloneq\left\{\hat{\Delta} :  
\begin{bmatrix} I \\ \hat{\Delta} \end{bmatrix}^\top
Y
\begin{bmatrix} I \\ \hat{\Delta} \end{bmatrix}
\succeq 0
\right\}
\end{equation}
with
\begin{equation*}
\begin{aligned}
    Y\!=\! \begin{bmatrix}
     \star
\end{bmatrix}^\top\!\!
\begin{bmatrix}
     \tilde{G}_{11}-\tilde{G}_{12} \tilde{G}_{22}^{-1}\tilde{G}_{12}^\top \!\!\!\!\!\!\!\!\!& 0 \\
     0 \!\!\!\!\!\!\!\!\!& (E\tilde{G}_{22}^{-1} E^\top)^{-1}
\end{bmatrix}\!\!
\begin{bmatrix}
     I \!\!\!\!\!& 0\\
     E\tilde{G}_{22}^{-1}\tilde{G}_{12}^\top \!\!\!\!\!& I
\end{bmatrix}.
\end{aligned}
\end{equation*}
\vspace{1pt}
\end{mylem}

\begin{proof}
Since $\tilde{G}_{22}\prec 0$ and $\tilde{G}_{11}-\tilde{G}_{12} \tilde{G}_{22}^{-1}\tilde{G}_{12}^\top \succ 0$, the QMI in $S_{E\tilde{\Delta}}$ can be equivalently rewritten in a block diagonal structure as
\begin{equation}
\begin{bmatrix} I \\ \tilde{\Delta}-\tilde{\Delta}_0 \end{bmatrix}^\top\!\!
\begin{bmatrix}
    \tilde{G}_{11}\!-\!\tilde{G}_{12} \tilde{G}_{22}^{-1}\tilde{G}_{12}^\top \!\!& 0 \\
               0 \!\!& \tilde{G}_{22}
\end{bmatrix}\!\!
\begin{bmatrix} I \\ \tilde{\Delta} -\tilde{\Delta}_0 \end{bmatrix}\succeq 0, \label{eq:LemmaE:QMITilde}
\end{equation}
where $\tilde{\Delta}_0 = -\tilde{G}_{22}^{-1}\tilde{G}_{12}^\top$. 
Since $E$ having full row rank and $\tilde{G}_{22} \prec 0$, $E\tilde{G}_{22}^{-1} E^\top$ is invertible and the QMI for $S_{\hat{\Delta}}$ is equivalent to
\begin{equation}
\begin{bmatrix} I \\ \hat{\Delta}\!\!-\!\!E\tilde{\Delta}_0 \end{bmatrix}^\top\!\!
\begin{bmatrix}
    \tilde{G}_{11}\!-\!\tilde{G}_{12} \tilde{G}_{22}^{-1}\tilde{G}_{12}^\top \!\!\!\!\!\!\!\!\!\!\!\!\!\!& 0 \\
               0 \!\!\!\!\!\!\!\!\!\!\!\!\!\!& (E\tilde{G}_{22}^{-1} E^\top)^{-1}
\end{bmatrix}\!\!
\begin{bmatrix} I \\ \hat{\Delta} \!\!-\!\!E\tilde{\Delta}_0 \end{bmatrix}\!\succeq\! 0.\label{eq:LemmaE:QMIhat}
\end{equation}
The idea of this proof is to show every component of $\tilde{\Delta}$ lying in the kernel of $E$ vanishes when considering $E\tilde{\Delta}$ and hence can be removed. 
Furthermore, we define $E^\perp$ to describe a matrix with rows forming a basis of the kernel of $E$. 
It follows directly that $EE^{\perp^\top}=0$.
Hence, it is possible to compute the matrix inverse
\begin{equation}\label{lemma1:proof1}
    \begin{bmatrix} E \\ \!E^\perp \tilde{G}_{22} \!\end{bmatrix}^{-1}\!\!\!\!\!\!=\!\!
    \begin{bmatrix} \tilde{G}_{22}^{-1}E^\top\!\! (E\tilde{G}_{22}^{-1} E^\top)^{-1} \!\!& E^{\perp\top}\!\!(E^\perp \tilde{G}_{22} E^{\perp\top})^{-1} \end{bmatrix} \!.
\end{equation}
Since for any invertible square matrix, the left and right inverse are identical and $\tilde{G}_{22}=I\tilde{G}_{22}I$, it holds that
\begin{equation}\label{lemma1:proof2}
    \tilde{G}_{22} \!=\! 
    \begin{bmatrix} \star  \end{bmatrix}^\top\!\!
    \begin{bmatrix}
        (E\tilde{G}_{22}^{-1}E^\top)^{-1} \!\!\!\!\!\!\!\!\!\!\!\!& 0 \\
        0 \!\!\!\!\!\!\!\!\!\!\!\!& (E^\perp \tilde{G}_{22} E^{\perp\top})^{-1}\\
    \end{bmatrix}\!\!
    \begin{bmatrix} E \\E^\perp \tilde{G}_{22}   \end{bmatrix}.
\end{equation}

Using \eqref{lemma1:proof1} and \eqref{lemma1:proof2}, we show equality of both sets by considering each set inclusion separately.

$S_{E\tilde{\Delta}}\subseteq S_{\hat{\Delta}}:$ Let $\tilde{\Delta}$ satisfy \eqref{eq:LemmaE:QMITilde}. 
Since $\tilde{G}_{22}$ and $E\tilde{G}_{22}^{-1} E^\top$ are invertible, we can represent $\tilde{\Delta}$ by
\begin{equation*}
\tilde{\Delta} \!\!=\!\! \begin{bmatrix} \tilde{G}_{22}^{-1}E^\top (E\tilde{G}_{22}^{-1} E^\top)^{-1} \!\!& E^{\perp\top}(E^\perp \tilde{G}_{22} E^{\perp\top})^{-1} \end{bmatrix}\!\!
\begin{bmatrix} \tilde{\Delta}_1 \\ \tilde{\Delta}_2 \end{bmatrix}
\end{equation*}
with $E\tilde{\Delta} = \tilde{\Delta}_1$ and $E^\perp \tilde{G}_{22}\tilde{\Delta}=\tilde{\Delta}_2$.
Inserting this in \eqref{eq:LemmaE:QMITilde} and employing \eqref{lemma1:proof2} for $\tilde{G}_{22}$ results in
\begin{equation*}
\begin{aligned}
\begin{bmatrix}
\star
\end{bmatrix}^\top\!\!
\begin{bmatrix}
\tilde{G}_{11}\!-\!\tilde{G}_{12} \tilde{G}_{22}^{-1}\tilde{G}_{12}^\top \!\!\!\!\!&0 \!\!\!\!\!\!\!\!\!& 0 \\
0            \!\!\!\!\!& (E\tilde{G}_{22}^{-1}E^\top)^{-1} \!\!\!\!\!\!\!\!\!& 0\\
0            \!\!\!\!\!&                                 0 \!\!\!\!\!\!\!\!\!& (E^\perp \tilde{G}_{22} E^{\perp\top})^{-1}\\
\end{bmatrix}\!\!\!\\
\begin{bmatrix}
I \\
\tilde{\Delta}_1 - E\tilde{\Delta}_0 \\
\tilde{\Delta}_2 \!-\! E^\perp \tilde{G}_{22}\tilde{\Delta}_0 \\
\end{bmatrix}\succeq 0.
\end{aligned}
\end{equation*}
Since $(E^\perp \tilde{G}_{22} E^{\perp\top})^{-1}\prec 0$, the above inequality results in \eqref{eq:LemmaE:QMIhat} by constructing $\hat{\Delta}=\tilde{\Delta}_1\in S_{\hat{\Delta}}$.

$S_{E\tilde{\Delta}}\supseteq S_{\hat{\Delta}}:$ Let $\hat{\Delta}$ satisfy \eqref{eq:LemmaE:QMIhat} and define the following candidate solution for $\tilde{\Delta}$
\begin{equation*}
\begin{aligned}
    \tilde{\Delta}\!\!=\!\! \tilde{G}_{22}^{-1}E^\top\!\!(\!E\tilde{G}_{22}^{-1}E^\top\!)^{-1}\!\hat{\Delta} 
    \!+\! E^{\perp\top}\!\!(E^\perp \tilde{G}_{22} E^{\perp\top}\!)^{-1}\!E^\perp \tilde{G}_{22}\tilde{\Delta}_0
\end{aligned}
\end{equation*}
with
$E\tilde{\Delta}=\hat{\Delta}$.
It remains to show that $\tilde{\Delta}$ satisfies \eqref{eq:LemmaE:QMITilde}.
Employing \eqref{lemma1:proof2} for $\tilde{G}_{22}$ and inserting the candidate solution $\tilde{\Delta}$ in \eqref{eq:LemmaE:QMITilde} yields \eqref{eq:LemmaE:QMIhat}, which is satisfied by construction.
\end{proof}

Lemma \ref{lemma1} is a dual version of \cite[Lemma 1]{braendle2024parameterization} and can be used to project a set described by a latent variable to an equal set with a lower dimensional latent variable. 
In the following, we state an assumption for the collected data.

\begin{assum}\label{assum2}\upshape
The collected data satisfy $z_i^d\neq 0$ for all $i\in\mathbb{I}_{[0, T-1]}$.
\end{assum}

In the following lemma, we use Lemma \ref{lemma1} to characterize the system matrices consistent with the data.

\begin{mylem}[Data-driven characterization of $\Sigma$]\upshape\label{lemma2}
Suppose Assumptions \ref{assumption:delta} and \ref{assum2} hold.
Then, the set $\Sigma$ is equal to
\begin{equation}\label{lemma:Sigma}
\left\{\!\!(A, B):\!\!\!\!\!\!
\begin{gathered}
\begin{bmatrix}
I \\A^\top \\B^\top
\end{bmatrix}^\top
M(\alpha)
\begin{bmatrix}
I \\A^\top \\B^\top
\end{bmatrix}\!\succeq\! 0,  \\
\forall \alpha\!=\!(\alpha_0, \ldots, \alpha_{T-1}), \alpha_i\!\geq\! 0, i\in\mathbb{I}_{[0, T-1]}
\end{gathered}
\right\},
\end{equation}
where
\begin{equation}\label{lemma:Pi}
M(\alpha)=\sum_{i=0}^{T-1}\alpha_i
\begin{bmatrix}
I &x_{i+1}^d\\
0 &-x_i^d\\
0 &-u_i^d
\end{bmatrix}
N_i
\begin{bmatrix}
I &x_{i+1}^d\\
0 &-x_i^d\\
0 &-u_i^d
\end{bmatrix}^\top,
\end{equation}
and $N_i$ is defined as 
\begin{equation}\label{Ni}
N_i\!\!=\!\!\begin{bmatrix}\star\end{bmatrix}^\top\!\!\!
\begin{bmatrix}
    G_{11}\!\!-\!\!G_{12}G_{22}^{-1}G_{12}^\top \!\!\!\!\!\!\!\!\!\!\!\!&0\\
    0 \!\!\!\!\!\!\!\!\!\!\!\!&({z_i^d}^\top G_{22}^{-1} z_i^d)^{-1}
\end{bmatrix}\!\!\!
\begin{bmatrix}
    I \!\!\!\!&0\\
    {z_i^d}^\top G_{22}^{-1}G_{12}^\top \!\!\!\!&I
\end{bmatrix}.
\end{equation}
\end{mylem}

\begin{proof}
Given data $z_i^d, \forall i\in\mathbb{I}_{[0, T-1]}$, we define 
\begin{equation}\label{lemma2:proof1}
\left\{
{z_i^d}^\top {\Delta_i^d}^\top:\begin{bmatrix}
    I\\
    {\Delta_i^d}^\top
\end{bmatrix}^\top
G
\begin{bmatrix}
    I\\
    {\Delta_i^d}^\top
\end{bmatrix}\succeq 0
\right\},
\end{equation}
which includes all possible realizations of ${z_i^d}^\top {\Delta_i^d}^\top$ given that Assumption \ref{assumption:delta} holds.
Since $z_i^d\neq 0$ for all $i\in\mathbb{I}_{[0, T-1]}$, by applying Lemma \ref{lemma1} for $\tilde{\Delta}={\Delta_i^d}^\top, E={z_i^d}^\top, \tilde{G}=G$, the set in \eqref{lemma2:proof1} is equal to
\begin{equation}\label{lemma2:proof2}
\left\{
{z_i^d}^\top {\Delta_i^d}^\top:\begin{bmatrix}
    I\\
    (\Delta_i z_i^d)^\top
\end{bmatrix}^\top
N_i
\begin{bmatrix}
    I\\
    (\Delta_i z_i^d)^\top
\end{bmatrix}\succeq 0
\right\},
\end{equation}
where $N_i$ is defined as in \eqref{Ni}.
The set in \eqref{lemma2:proof2} includes all possible realization of ${z_i^d}^\top {\Delta_i^d}^\top$ with $\Delta_i^d\in\Pi$.
Based on the system dynamics \eqref{system}, $(A, B)\in\Sigma_i$ if and only if there exists $\Delta_i^d\in\Pi$ such that the following equation holds
\begin{equation}\label{lemma:sigma:proof3}
\Delta_i^d z_i^d=x_{i+1}^d-Ax_i^d-Bu_i^d.
\end{equation}
Thus, replacing $\Delta_i^d z_i^d$ by $x_{i+1}^d-Ax_i^d-Bu_i^d$ in \eqref{lemma2:proof2}, the following set includes all possible $(A, B)$ that are consistent with the data $x_i^d, u_i^d, z_i^d, x_{i+1}^d$ for all $i\in\mathbb{I}_{[0, T-1]}$
\begin{equation}\label{lemma:sigma1}
\left\{\!\!(A, B)\!\!:\!\!\!
\begin{bmatrix}
I\\ A^\top \\B^\top
\end{bmatrix}^\top\!\!\!
\begin{bmatrix}
I \!\!&x_{i+1}^d\\
0 \!\!&-x_i^d\\
0 \!\!&-u_i^d
\end{bmatrix}
\!\!N_i\!\!
\begin{bmatrix}
I \!\!&x_{i+1}^d\\
0 \!\!&-x_i^d\\
0 \!\!&-u_i^d
\end{bmatrix}^\top\!\!\!
\begin{bmatrix}
I\\ A^\top \\B^\top
\end{bmatrix}
\!\!\succeq \!0\!
\right\}.
\end{equation}
Analogous to \cite{berberich2023combining, bisoffi2021trade}, and using the characterization of $\Sigma_i$ in \eqref{lemma:sigma1}, we obtain the set $\Sigma$ as in \eqref{lemma:Sigma}.
\end{proof}


\begin{remark}\upshape
In the literature \cite{verhoek2024decoupling, berberich2023combining,miller2022data}, data-driven parametrizations of the consistent system matrices of an LPV system were derived under the assumption that the scheduling signal is known.
Lemma 2 contains a parameterization for unknown scheduling signals satisfying a given QMI.
\end{remark}

\section{Data-Driven Min-Max MPC for LPV systems}\label{sec:4}
In Section \ref{sec:4.1}, we propose a data-driven min-max MPC problem for LPV systems, which we reformulate as a semi-definite programming (SDP) problem.
In Section \ref{sec:4:2}, we show that the SDP problem is recursively feasible and the resulting closed-loop system is exponentially stable and satisfies the constraints for all consistent system matrices and all scheduling signals satisfying Assumption \ref{assumption:delta}.

\subsection{Data-driven min-max MPC problem}\label{sec:4.1}
Given the current state $x_t\in\mathbb{R}^{n_x}$, the data-driven min-max MPC problem is formulated as follows:
\begin{subequations}\label{mpc}
\begin{align}
J^\star&(x_t)=\min_{\bar{u}(t)}\!\!\!\max_{(A, B)\in \Sigma\atop\Delta_{t+k}\in\Pi,k\in\mathbb{N}}\sum_{k=0}^\infty \|\bar{u}_k(t)\|_R^2+\|\bar{x}_k(t)\|_Q^2\label{mpc:cost}\\
\text{s.t. }&\bar{x}_{k+1}(t)=A\bar{x}_k(t)+B\bar{u}_k(t)+\Delta_{t+k}\bar{z}_k(t),\label{mpc:con1}\\
&\bar{x}_0(t)=x_t,\label{mpc:con2}\\
&\|\bar{u}_k(t)\|_{S_u}\leq 1, \forall k\in\mathbb{N},\label{mpc:con3}\\
&\|\bar{x}_k(t)\|_{S_x}\leq 1, \forall (A, B)\in\Sigma, \Delta_{t+k}\in \Pi, k\in\mathbb{N}.\label{mpc:con4}
\end{align}
\end{subequations}
In problem \eqref{mpc}, $\bar{x}_k(t)$ and $\bar{u}_k(t)$ denote the predicted state and input at time $t+k$ given the current state $x_t$.
The optimization variable $\bar{u}(t)$ is a sequence of predicted inputs.
We initialize $\bar{x}_0(t)$ as the current state $x_t$ in constraint \eqref{mpc:con2}.
The objective of problem \eqref{mpc} is to minimize the worst-case infinite-horizon cost over all consistent system matrices in $\Sigma$ and all scheduling signals $\Delta_{t+k}\in\Pi, k\in\mathbb{N}$.
The problem \eqref{mpc} is intractable because of the min-max formulation and infinite-horizon cost.
Similar to \cite{xie2024minmax,kothare1996robust}, we restrict our attention to a state-feedback control law $\bar{u}_k(t)=F_t\bar{x}_k(t), \forall k\in\mathbb{I}_{[0, \infty)}$.
This allows us to transform problem \eqref{mpc} into a computationally tractable SDP problem. 

\begin{figure*}[htbp]
\begin{subequations}\label{SDP}
\begin{align}
    &\minimize\limits_{\gamma>0, H\in\mathbb{R}^{n_x\times n_x}, L\in\mathbb{R}^{n_u\times n_x}, \alpha\in\mathbb{R}^{T}, \lambda>0}\gamma\label{SDP:obj}\\
    \text{s.t. }
    &\begin{bmatrix}1 &x_t^\top\\
x_t &H\end{bmatrix}\succeq 0, \begin{bmatrix}
H &L^\top\\ L &S_u^{-1}
\end{bmatrix}\succeq 0, \begin{bmatrix}
H &H^\top\\ H &S_x^{-1}
\end{bmatrix}\succeq 0,\alpha=(\alpha_0, \ldots, \alpha_{T-1}), \alpha_i\geq 0, \forall i\in\mathbb{I}_{[0, T-1]},\label{SDP:con1}\\
    &\begin{bmatrix}
\begin{bmatrix}
H-\lambda G_{11} &0 &0 &-\lambda G_{12}\\
0 &0 &0 &0\\
0 &0 &0 &0\\
-\lambda G_{12}^\top &0 &0 &-\lambda G_{22}
\end{bmatrix}-
\begin{bmatrix}
M(\alpha) &0\\
0 &0
\end{bmatrix}
&\begin{bmatrix}
\begin{bmatrix}
0\\H\\L
\end{bmatrix} &0 &\begin{bmatrix}
0\\H\\L
\end{bmatrix} &0\\
CH+DL &0 &CH+DL &0
\end{bmatrix}\\
\begin{bmatrix}
\begin{bmatrix}
    0 &H^\top &L^\top
\end{bmatrix} &(CH+DL)^\top\\
0 &0\\
\begin{bmatrix}
    0 &H^\top &L^\top
\end{bmatrix} &(CH+DL)^\top\\
0&0
\end{bmatrix}
&\begin{bmatrix}
H &\Phi^\top &0 &0\\
\Phi &\gamma I &0 &0\\
0 &0 &H &\Phi^\top\\
0 &0 &\Phi &\gamma I
\end{bmatrix}
\end{bmatrix}\succ 0,\label{SDP:con2}
\end{align}
\end{subequations}
\vspace{-30pt}
\end{figure*}

Given the state $x_t\in\mathbb{R}^{n_x}$ and the input-state data \eqref{data}, the 
SDP problem is formulated as in \eqref{SDP}, where $\Phi=\begin{bmatrix}M_R L\\ M_Q H\end{bmatrix}$, $M_R^\top M_R=R$ and $M_Q^\top M_Q=Q$.
The optimal solution of problem \eqref{SDP} depends on the current measured state $x_t$.
Therefore, we denote it by $\gamma_{x_t}^\star, H_{x_t}^\star, L_{x_t}^\star, \alpha_{x_t}^\star, \lambda_{x_t}^\star$.
The corresponding optimal state-feedback gain is given by $F_{x_t}^\star=L_{x_t}^\star (H_{x_t}^\star)^{-1}$.

In the following theorem, we show that the SDP problem \eqref{SDP} minimizes an upper bound on the optimal cost of the problem \eqref{mpc}. 

\begin{mythm}\upshape\label{Theorem1}
Suppose that Assumption \ref{assumption:delta} and \ref{assum2} hold. 
If there exist $\gamma, \lambda>0$, $H\in\mathbb{R}^{n_x\times n_x}$, $ L\in\mathbb{R}^{n_u\times n_x}$, $\alpha\in\mathbb{R}^T$ such that the inequalities in \eqref{SDP:con1}-\eqref{SDP:con2} hold, then the optimal cost of \eqref{mpc} is guaranteed to be upper bounded by $\gamma$.
\end{mythm}
\begin{proof}
We define a quadratic Lyapunov function $V(x)=\|x\|_P^2$ for $x\in\mathbb{R}^{n_x}$, where $P\succ 0$.
Suppose that $V$ satisfies the following inequality 
\begin{equation}\label{Theorem1:proofV}
V(\bar{x}_{k+1}(t))-V(\bar{x}_k(t))\leq -\|\bar{u}_k(t)\|_R^2-\|\bar{x}_k(t))\|_Q^2
\end{equation}
for all $\bar{x}_k(t)\in\mathbb{R}^{n_x}, \bar{u}_k(t)=F\bar{x}_k(t)$ and $\bar{x}_{k+1}(t)$ predicted by \eqref{mpc:con1} with any $(A, B)\in \Sigma$ and $\Delta_{t+k}\in\Pi, k\in\mathbb{N}$.
Using the same arguments as in \cite[equations (5)-(7)]{xie2024minmax}, $V(x_t)$ is an upper bound on the cost \eqref{mpc:cost} if \eqref{Theorem1:proofV} holds.
Using the Schur complement, $\|x_t\|_P^2\leq \gamma$ is equivalent to the first LMI in \eqref{SDP:con1}.
The input and state constraint \eqref{mpc:con3}-\eqref{mpc:con4} can be reformulated as the second and third LMIs in \eqref{SDP:con1}, which can be shown analogously to \cite{xie2024minmax}.
Therefore, if \eqref{SDP:con1} and \eqref{Theorem1:proofV} hold, the optimal cost of \eqref{mpc:cost}-\eqref{mpc:con2} is upper bounded by $\gamma$, i.e., $J^\star(x_t)\leq V(x_t)\leq\gamma$. 

In the following, we prove that  \eqref{Theorem1:proofV} holds.
We first apply the Schur complement to the two blocks $\begin{bmatrix}H &\Phi^\top\\\Phi &\gamma I\end{bmatrix}$ in \eqref{SDP:con2}, leading to a different matrix with $\Omega=\begin{bmatrix}H-\frac{1}{\gamma}\Phi^\top\Phi &0\\ 0 &H-\frac{1}{\gamma}\Phi^\top\Phi\end{bmatrix}$ in the right lower block.
Then applying the Schur complement to the resulting matrix by inverting $\Omega$, we obtain
\begin{subequations}
\begin{align}
&S-\begin{bmatrix}
M(\alpha) &\!\!\!\!0\\
0 &\!\!\!\!0
\end{bmatrix}\!-\!\lambda\!\begin{bmatrix}
G_{11}&\!0 &\!0 &\!G_{12}\\
0 &\!0 &\!0 &\!0\\
0 &\!0 &\!0 &\!0\\
G_{12}^\top &\!0 &\!0 &\!G_{22}
\end{bmatrix}\!\succ\! 0,\label{Theorem1:proof11}\\
&H-\frac{1}{\gamma}\Phi^\top \Phi\succ 0,\label{Theorem1:proof12}
\end{align}
\end{subequations}
where $S$ is defined as in \eqref{Theorem1:proof2}.
\begin{figure*}[htbp]
\begin{equation}\label{Theorem1:proof2}
\begin{bmatrix}
I\\
A^\top\\
B^\top\\
\!\Delta_{t+k}^\top\!
\end{bmatrix}^\top\!
\begin{matrix}\underbrace{\begin{bmatrix}
\begin{array}{c|c|c}
H &0 &0\\ \hline 
0 &-\begin{bmatrix}H\\L\end{bmatrix}\!(H\!-\!\frac{1}{\gamma}\Phi^\top\Phi)^{-1}\!\begin{bmatrix}H\\L\end{bmatrix}^\top &-\begin{bmatrix}H\\  L\end{bmatrix}\!(H\!-\!\frac{1}{\gamma}\Phi^\top\Phi)^{-1}\!(CH+DL)^\top\\ \hline 
0 &-(CH+DL)(H\!-\!\frac{1}{\gamma}\Phi^\top\Phi)^{-1}\begin{bmatrix}H\\ L\end{bmatrix}^\top &-(CH+DL)(H\!-\!\frac{1}{\gamma}\Phi^\top\Phi)^{-1}(CH+DL)^\top
\end{array}
\end{bmatrix}}\\S\end{matrix}\!
\begin{bmatrix}
I\\
A^\top\\
B^\top\\
\!\Delta_{t+k}^\top\!
\end{bmatrix}\!\succ\! 0.
\end{equation}
\vspace{-30pt}
\end{figure*}
Combining \eqref{assum:delta} and \eqref{lemma:Sigma}, for any $\Delta_{t+k}\in\Pi, k\in\mathbb{N}$ and $(A, B)\in\Sigma$ the following inequality is satisfied with any $\lambda\geq 0$ and any $\alpha=(\alpha_0, \ldots, \alpha_{T-1}), \alpha_i\geq 0, i\in\mathbb{I}_{[0, T-1]}$
\begin{equation}\label{Theorem1:proof3}
\begin{bmatrix}
I\\
A^\top\\
B^\top\\
\!\Delta_{t+k}^\top\!
\end{bmatrix}^\top\!\!\!\!\!
\left(\!\!\!\begin{bmatrix}
M(\alpha) &\!\!\!\!0\\
0 &\!\!\!\!0
\end{bmatrix}\!\!+\!\!\lambda\!\!\begin{bmatrix}
G_{11}&\!\!\!0 &\!\!0 \!\!&G_{12}\\
0 &\!\!\!0 &\!\!0 \!\!&0\\
0 &\!\!\!0 &\!\!0 \!\!&0\\
G_{12}^\top &\!\!0 &\!0 \!\!&G_{22}
\end{bmatrix}\!\right)\!\!\!
\begin{bmatrix}
I\\
A^\top\\
B^\top\\
\!\Delta_{t+k}^\top\!
\end{bmatrix}\!\!\!\succeq\! 0.
\end{equation}

Pre-multiplying \eqref{Theorem1:proof11} with $\begin{bmatrix}I \!\!&A \!\!&B \!\!&\Delta_{t+k}\end{bmatrix}$ and post-multiplying with $\begin{bmatrix}I \!\!&A \!\!&B \!\!&\Delta_{t+k}\end{bmatrix}^\top$, the resulting inequality and \eqref{Theorem1:proof3} imply that \eqref{Theorem1:proof2} holds for any $\Delta_{t+k}\in\Pi, k\in\mathbb{N}$ and any $(A, B)\in\Sigma$.
The inequality \eqref{Theorem1:proof2} is equivalent to
\begin{equation}\label{Theorem1:proof4}
H\!-\![\star](H\!-\!\frac{1}{\gamma}\Phi^\top\Phi)^{-1}
[(AH\!+\!BL)\!+\!\Delta_{t+k}(CH\!+\!DL)]^\top\!\succ\! 0.
\end{equation}

Applying the Schur complement twice, \eqref{Theorem1:proof4} is equivalent to the following inequality and \eqref{Theorem1:proof12}
\begin{equation}\label{Theorem1:proof5}
\begin{aligned}
[\star]^\top\!\! H^{-1}\![(AH\!+\!BL)\!+\!\Delta_{t+k}(CH\!+\!DL)]\!\!-\!\!H\!\!+\!\!\frac{1}{\gamma}\Phi^\top\!\Phi\!\prec\! 0.
\end{aligned}
\end{equation}
We define $P=\gamma H^{-1}$ and $F=L H^{-1}$. Multiplying both sides of \eqref{Theorem1:proof5} with $P$ and dividing the resulting inequality by $\gamma$, we obtain that
\begin{equation}\label{Theorem1:proofbasicin}
\begin{aligned}
[A\!+\!BF\!+\!\Delta_{t+k}(C\!+\!DF)]^\top P [A\!+\!BF\!+\!\Delta_{t+k}(C\!+\!DF)]
\\-P+Q+F^\top R F\prec 0,
\end{aligned}
\end{equation}
holds for any $(A, B)\in\Sigma$ and any $\Delta_{t+k}\in\Pi, k\in\mathbb{N}$.
Multiplying \eqref{Theorem1:proofbasicin} from left and right by $\bar{x}_k(t)$ and $\bar{x}_k(t)^\top$ and substituting $\bar{u}_k(t)=F\bar{x}_k(t)$ into the resulting inequality, we obtain that  \eqref{Theorem1:proofV} holds for any $(A, B)\in \Sigma$ and any $\Delta_{t+k}\in\Pi, k\in\mathbb{N}$.
Together with the previous arguments, we conclude that the optimal cost of \eqref{mpc:cost}-\eqref{mpc:con2} is guaranteed to be upper bounded by $\gamma$ if \eqref{SDP:con1}-\eqref{SDP:con2} hold.
\end{proof}

In order to control the unknown LPV system \eqref{system}, the SDP problem \eqref{SDP} is solved in a receding-horizon manner, where the optimal state-feedback gain is recomputed at each time step, as detailed in Algorithm~1. 

\begin{algorithm}[htb]
\caption{\!Data-driven min-max MPC scheme.\!\!\!}
    \SetKwData{Left}{left}
    \SetKwData{Up}{up}
    \SetKwFunction{FindCompress}{FindCompress}
    \nl At time $t=0$, measure state $x_0$\;
    \nl Solve the problem \eqref{SDP}\;
    \nl Apply the input $u_t=F_{x_t}^\star x_t$\;
    \nl Set $t=t+1$,  measure state $x_t$ and go back to 3\;
    \label{algorithm1}
\end{algorithm}
\vspace{-10pt}

\subsection{Closed-loop guarantees}\label{sec:4:2}
In the following theorem, we prove that the SDP problem \eqref{SDP} is recursively feasible, the closed-loop system resulting from the Algorithm~1 is exponentially stabilized to the origin and satisfies the constraints.
\begin{mythm}\label{Theorem2}\upshape
Suppose Assumption \ref{assumption:delta} and \ref{assum2} holds. If the optimization problem \eqref{SDP} is feasible at time $t=0$, then
\begin{enumerate}[(1)]
\item it is feasible at any time $t\in\mathbb{N}$;
\item the origin is exponentially stabilized for the closed-loop system;
\item the closed-loop trajectory satisfies the input and state constraints.
\end{enumerate}
\end{mythm}
\begin{proof}
The proof is analogous to that of Lemma 1, Theorem 2 and 3 in \cite{xie2024minmax}, so we only provide a sketch of proof here.
The state at time $t+1$ is given by $x_{t+1}=(A_s+B_sF_{x_t}^\star)x_t+\Delta_t(C+DF_{x_t}^\star)x_t$.
Since $F_{x_t}^\star$ and $P_{x_t}^\star$ are derived from the optimal solution of problem \eqref{SDP}, replacing $F$ and $P$ with $F_{x_t}^\star$ and $P_{x_t}^\star$, the inequality \eqref{Theorem1:proofbasicin} holds for any $(A, B)\in \Sigma$ and any $\Delta_{t+k}\in\Pi, k\in\mathbb{N}$ as proved in Theorem~1.
Since $(A_s, B_s)\in\Sigma$ and $\Delta_t\in\Pi$, setting $k=0$ in inequality \eqref{Theorem1:proofbasicin} and multiplying both sides of the resulting inequality with $x_t^\top$ and $x_t$, we obtain 
$\|x_{t+1}\|_{P_{x_t}^\star}^2\!-\|x_t\|_{P_{x_t}^\star}^2 \!\leq\! -x_t^\top(Q+F_{x_t}^\star R F_{x_t}^\star)x_t\!\leq\! -\lambda_{\min}(Q)\|x_t\|^2$,
where $\lambda_{\min}(Q)$ is the minimal eigenvalue of the matrix $Q$.
Since $\|x_{t+1}\|_{P_{x_t}^\star}^2\leq \|x_t\|_{P_{x_t}^\star}^2 \leq \gamma_{x_t}^\star$, the optimal solution at time $t$ is a feasible solution of the problem \eqref{SDP} at time $t+1$.
Since $P_{x_t}^\star$ is a feasible solution while $P_{x_{t+1}}^\star$ is the optimal solution of the problem \eqref{SDP} given the state $x_{t+1}$, we have $\|x_{t+1}\|_{P_{x_{t+1}}^\star}^2\leq \|x_{t+1}\|_ {P_{x_t}^\star}^2$.
Choosing the Lyapunov function as $V(x_t)=\|x_t\|_{P_{x_t}^\star}^2$, which is lower bounded by $\|x_t\|_Q^2$ and upper bounded by $\|x_t\|_{P_{x_0}^\star}^2$, we have $V(x_{t+1})-V(x_t)\leq -\lambda_{\min}(Q)\|x_t\|^2$.
This proves that the origin is exponentially stabilized in closed loop.
The proof for constraint satisfaction proceeds in two steps: 1) proving that the set $\mathcal{E}=\{x\in\mathbb{R}^n:x^\top P x \leq \gamma\}$ is a robust positive invariant set for the system 
with $(A, B)\in\Sigma$ and $\Delta_t\in\Pi$;
2) showing that if the second the third LMIs in \eqref{SDP:con1} hold, then all state $x\in\mathcal{E}$ and the input $u=F_x^\star x, x\in\mathcal{E}$ satisfies the state and input constraints. 
These steps are analogous to \cite{xie2024minmax}.
\end{proof}


\begin{remark}\upshape
Our framework can be extended to process noise in the system dynamics, i.e., $x_{t+1}=A_sx_t+B_su_t+\omega_t+d_t$, where $d_t$ is bounded process noise, following a similar approach in \cite{xie2024minmaxrobust}.
By assuming that $d_t$ satisfies a QMI, a combined QMI on $\omega_t+d_t$ can be constructed using the bound on $d_t$ as well as Assumption~\ref{assumption:delta}.
With this QMI, the consistent system matrices can be characterized analogous to the proposed approach.
Moreover, we conjecture that also process noise during online operation can be considered, in which case robust stability of a robust positive invariant set around the origin can be shown, compare \cite{xie2024minmaxrobust}. 
\end{remark}

\section{Simulation}\label{sec:5}
In this section we apply the proposed approach in simulation.
We consider the discrete-time angular positioning system from \cite{kothare1996robust}
\begin{equation}\label{system:simulation}
x_{t+1}=\begin{bmatrix}
    1 &0.1\\
    0 &1-0.1a_t
\end{bmatrix}x_t+
\begin{bmatrix}
    0\\ 0.0787
\end{bmatrix}u_t,
\end{equation}
where $x_t$ represents the angular position and the angular velocity of the antenna, $u_t$ is the motor voltage and $a$ is proportional to the friction in the rotating parts of the antenna.
This system can be formulated as in \eqref{system} via 
\begin{equation}\nonumber
\begin{aligned}
A_s\!\!=\!\!\begin{bmatrix}1 \!\!\!\!&0.1\\0 \!\!\!\!&1\end{bmatrix}, B_s\!\!=\!\!\begin{bmatrix}0\\0.0787\end{bmatrix}\!,\!
C\!\!=\!\!\begin{bmatrix}
0 \!\!\!\!\!&0\\
0 \!\!\!\!\!&-0.1
\end{bmatrix}, D\!\!=\!0,\!\Delta_t\!\!=\!\!\begin{bmatrix}a_t \!\!\!\!\!&0\\0 \!\!\!\!\!&a_t\end{bmatrix}.
\end{aligned}
\end{equation}
In the following, we consider $a_t\in[0.05, 0.05+0.05c]$ with varying values of $c$ in the range $0<c\leq 2$ in order to study the influence of $c$ on the closed-loop performance. 
Each value of $c$ results in a corresponding QMI for $\Delta_t$ and produces the following matrices defined in \eqref{assum:delta}
\begin{equation}\nonumber
    \begin{aligned}
        &G_{11}\!=\!-0.25(1+c)I,G_{22}\!=\!-100I, 
G_{12}\!=\!(5+2.5c)I.
    \end{aligned}
\end{equation}
Several sequences of input-state data of length $T=20$ are generated for different values of $c$, where the scheduling signal is sampled uniformly in $[0.05, 0.05+0.05c]$ and the input is sampled uniformly in $[-1, 1]$.
The weight matrices of the data-driven min-max MPC scheme are chosen as $Q= I, R=0.01I$.
The matrices for constraints are chosen as $S_x=4I$ and $S_u=0.16$.
The initial state is chosen as $x_0=[0.05,0]$.

We apply the proposed data-driven min-max MPC scheme to the system \eqref{system:simulation} using the generated sequences of data.
All the closed-loop trajectories converge to the origin after a few iterations and the input and state constraints are satisfied.
Fig.~1(a) illustrates the closed-loop cost over $100$ iterations with different values of $c$. 
As $c$ increases, the closed-loop cost increases due to the increasing robustness requirements.
Furthermore, we study the influence of data length $T$ on the closed-loop cost. 
We set $c=1$ and vary the data length $T$, while other parameters including the constraints and weight matrices remain the same.
Fig.~1(b) shows the closed-loop cost over 100 iterations with different values of $T$.
When $T\leq 2$, the SDP problem \eqref{SDP} is infeasible. 
As the data length increases, the closed-loop cost decreases.

\begin{figure}
    \centering
    \subfigure[]{\label{pic:state1_result}
    \includegraphics[width=0.45\textwidth]{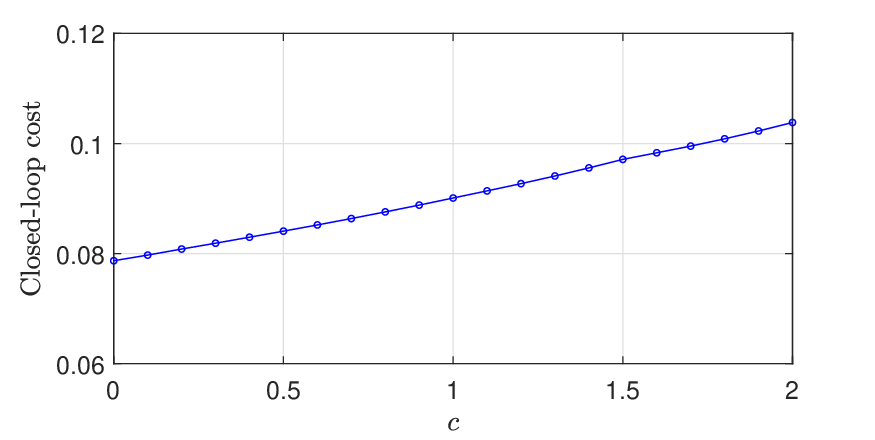}}
    \subfigure[]{\label{pic:state2_result}
    \includegraphics[width=0.45\textwidth]{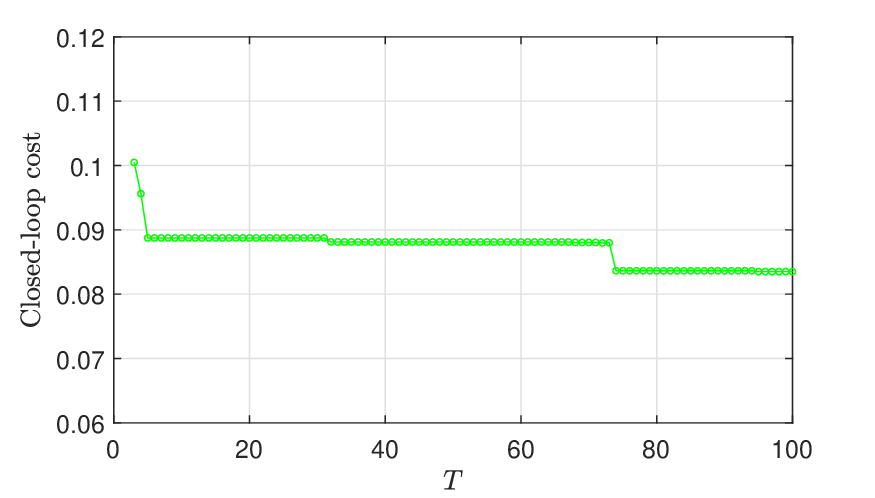}}
    \caption{Closed-loop cost over $100$ iterations under the proposed data-driven min-max MPC scheme with different values of the bound $c$ (subfigure (a)) and the data length $T$ (subfigure (b)).}
    \label{pic:simu_LTI}
    \vspace{-30pt}
\end{figure}

\section{Conclusion}\label{sec:6}
In this paper, we proposed a data-driven min-max MPC scheme for LPV systems with unknown scheduling signals.
Based on a novel data-driven characterization of the consistent system matrices set, an SDP was formulated to minimize an upper bound on the data-driven min-max MPC problem.
The SDP leads to a state-feedback gain that stabilizes the closed-loop system for any scheduling signal satisfying the QMI and any consistent system matrices.
The simulation results show the effectiveness of the proposed scheme.
In the future, we aim to reduce conservatism by incorporating estimates of the scheduling signals and uncertainty bounds.

\bibliographystyle{unsrt}
\bibliography{main}
 
\end{document}